\documentclass[10pt]{article}

\usepackage[T1]{fontenc}
\usepackage[utf8]{inputenc}
\usepackage[english]{babel}
\usepackage{amsmath,latexsym,amsfonts,amsthm,bm,mathtools}
\usepackage{amssymb}
\usepackage{graphicx}
\usepackage[top=1.5in, bottom=1.5in, left=1.25in, right=1.25in]{geometry}
\usepackage{booktabs,url,tikz}
\usepackage[makeroom]{cancel}
\usepackage{dsfont}
\usepackage{relsize}
\usepackage{multirow}

\usepackage{algorithm}
\usepackage{algpseudocodex}
\usepackage{bbm}
\usepackage{subcaption}

\theoremstyle{definition}
\newtheorem{definition}{Definition}

\theoremstyle{plain}
\newtheorem{proposition}{Proposition}

\theoremstyle{remark}
\newtheorem{remark}{Remark}

\newcommand{\mathDef}{\overset{\mathrm{def}}{=}}

\title{Non-ignorable fuzziness in granular counts: \\the case of RNA-seq data}
\author{Antonio Calcagn\`{i}$^{1\ast}$, Arianna Consiglio$^{2}$, Przemys{\l}aw Grzegorzewski$^{3}$, Corrado Mencar$^{4}$ \\\\
		\footnotesize{\sl $^{1}$ University of Padova, \sl $^{2}$ National Research Council, \sl $^{3}$ Warsaw University of Technology, \sl $^{4}$ University of Bari ``A. Moro''} \\
		\footnotesize{$\ast$ E-mail: antonio.calcagni@unipd.it}
	}

\date{}

\begin{document}

\maketitle

\begin{abstract}
RNA-seq count data are often affected by read-to-gene alignment ambiguity, especially in high-dimensional transcriptomics. This type of ambiguity can be conveniently expressed through granular counts, namely fuzzy-valued observations of latent discrete quantities. We study a class of fuzzy-reporting mechanisms and show that, when reporting exploits graded membership, ignorability fails generically, leading to a coarsening-not-at-random structure. A hierarchical model is then introduced as a tractable instance of this construction and illustrated using RNA-seq data.\\

\textsc{Supplementary Material}: Additional results are available in the supplementary material accompanying this submission ($\hookrightarrow$ \texttt{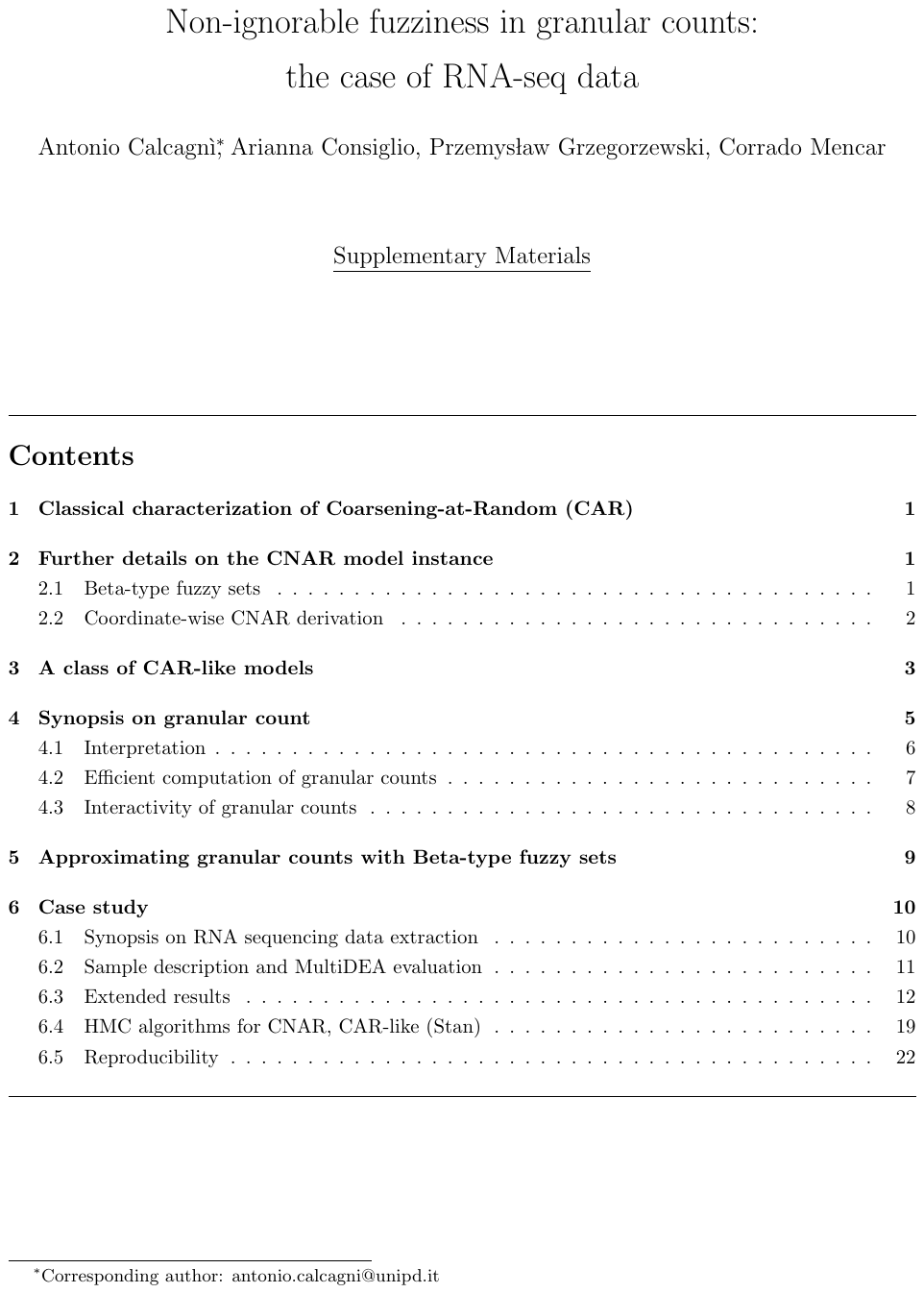}).
\end{abstract}

\vspace{2cm}

\section{Introduction}\label{sec1}

RNA-seq provides a natural motivating example for the statistical analysis of ambiguous count data. In high-dimensional transcriptomic settings, short reads are often compatible with multiple genes or isoforms, so that read-to-gene assignments are not always uniquely determined. While several strategies have been proposed to handle multireads, the resulting ambiguity is often treated as a technical problem rather than as a form of uncertainty intrinsic to the alignment itself \cite{ji2011bm}. In fact, this type of ambiguity can be naturally viewed as epistemic uncertainty, reflecting limitations in information and representation (for instance, because of shared exonic structure, sequence similarity, polymorphisms, incomplete annotation), and it cannot be modeled as measurement noise without losing some of its features \cite{consiglio2016fuzzy,deshpande2023rna}. From this viewpoint, ambiguous read assignment constitutes a form of \textit{granular counting}, resulting in fuzzy counts over competing loci or transcripts \cite{consiglio2016fuzzy,mencar2020granular}.

Although RNA-seq provides the main motivating example, similar issues arise whenever a latent quantity is observed through ambiguous allocation or graded compatibility with multiple alternatives. This is the case, for instance, in pooled testing, where multiple underlying states (e.g., different combinations of positives) may lead to the same observed outcome, and in multi-target tracking, where observations must be associated with competing targets. In all such cases, the resulting uncertainty can be formally represented by fuzzy counts $\tilde y$, defined via a possibility distribution $\xi_{y}:\mathbb N_0\to[0,1]$. By studying a class of fuzzy-reporting mechanisms linking an underlying precise count $Y\sim \mathcal F_{Y}(y;\boldsymbol\theta)$ to its fuzzy counterpart, we provide the main theoretical result: whenever reporting genuinely exploits graded membership, the induced mechanism is generically non-ignorable. Specifically, it behaves as a \textit{coarsening-not-at-random} (CNAR) process rather than a \textit{coarsening-at-random} one \cite{gill1997coarsening}, thus justifying a hierarchical approach where the observed fuzzy count is treated as an imprecise realization of the latent count. 

The remainder of this paper is organized as follows. Section~\ref{sec2} introduces basic notation and tools used throughout the paper. Section~\ref{sec3} shows that the fuzzy-reporting mechanism behaves as CNAR and motivates a hierarchical model for granular counts. Section~\ref{sec4} presents a real data application involving RNA-seq data and Section~\ref{sec5} concludes the paper by summarizing its main findings. Throughout the paper, references labeled S\textit{X}, or S\textit{X.Ya} (e.g., Figure S1, Section S6.3a) refer to the Supplementary Materials.

\section{Preliminaries}\label{sec2}

This section introduces the main definitions, notation, and technical tools used throughout the paper. In what follows, $S \mathDef \mathbb{N}_0$ denotes the state space of non-negative integer counts, $\mathcal{S} \mathDef \mathcal{P}(S)$ its power set, and $S_K \mathDef \{0,1,\ldots,K\} \subset S$ the truncated count space up to level $K$, for some fixed $K \geq 0$.

\begin{definition}
\textit{Fuzzy set.} A fuzzy subset $\tilde A$ of $S$ is specified by its membership function $\xi:S\to [0,1]$, where $\xi(y)$ quantifies the degree to which $y\in \tilde A$. The support of a fuzzy subset is $\text{supp}(\tilde A) \mathDef \{y\in S: \xi(y)>0\}$ while $\text{core}(\tilde A) \mathDef \{y\in S: \xi(y)=1\}$ is its core. In general, for $\alpha\in]0,1]$, the set $ A_\alpha \mathDef \{y\in S: \xi(y)\geq\alpha\}$ is the $\alpha$-cut of $\tilde A$. We assume that $\xi$ is {normalized}, i.e. $\sup_{y\in S}\xi(y)=1$, and all the membership functions are meant to be $\mathcal S$-measurable. There are several parametric families for specifying membership functions $\xi$ (e.g., triangular, trapezoidal) and, among them, the {beta-type} family (see Section S2.1) provides a flexible unimodal shape on a bounded support and admits an interpretable parameterization in terms of location $c\in [0,K]$ and precision $h>0$ \cite{calcagni2025bayesianize}. 
\end{definition}

\begin{remark}
When a sample of fuzzy observations $\{\tilde y_i\}_{i=1}^n$ is available, the parameters of a Beta-type fuzzy set $c$ and $h$ assume the role of statistics of the data (not to be confused with the statistical model parameters). 
\end{remark}

\begin{definition}
\textit{Statistical experiment.} $Y:(\Omega,\mathcal A,\mathbb P)\to(S,\mathcal S)$ is an $(\mathcal A-\mathcal S)$ measurable map (a random variable). For $\theta\in\Theta$, $\mathbb P_\theta$ denotes the induced distribution of $Y$ on $(S,\mathcal S)$, with $(\mathbb P_\theta)_{\theta\in\Theta}$ being a parametric family of probability measures on $(S,\mathcal S)$. The triple $(S,\mathcal S,\mathbb P_\theta)$ defines the usual statistical experiment.
\end{definition}

\begin{definition}
\textit{Space of bounded and measurable functions.} Let \\$\mathcal B_b(S,\mathcal S)\mathDef\{f:S\to\mathbb R \mid f \text{ is }\mathcal S\text{-measurable},~ \sup_{y\in S}|f(y)| <\infty\}$. Given a probability measure $\mathbb P_\theta$ on $(S,\mathcal S)$, the functional $C_\theta:\mathcal B_b(S,\mathcal S)\to\mathbb R$ defined by $C_\theta(f)\mathDef \sum_{y\in S} f(y)\mathbb P_\theta[Y=y]$ is linear and positive. The subset $M \mathDef \{ \xi \in \mathcal B_b(S,\mathcal S) \mid 0 \le \xi \le 1,~ \sup_{y\in S}\xi(y) = 1 \}$ is a normalized slice of the positive cone of $\mathcal B_b(S,\mathcal S)$, which is closed under $\vee$ (pointwise maximum), i.e. $(\xi \vee \xi')(y) = \max\{\xi(y),\xi'(y)\}$. If equipped with a $\sigma$-algebra $\mathcal M$ -- for instance, the cylindrical one generated by the evaluation maps $e_y:M\to[0,1]$, $e_y(\xi)\mathDef \xi(y)$-- $(M,\mathcal M)$ is a measurable space.
\end{definition}

\begin{definition}
\textit{Fuzzy sets \`a la Le Cam}. If $\{\xi_i\}_{i=1}^n \subseteq M$, then $M\subset\mathcal B_b(S,\mathcal S)$ is naturally framed within Le Cam's single-stage experiment \cite{gil1993statistical}, with $M$ playing the role of a class of measurable membership functions. In this setting, $C_\theta(\xi) \mathDef \sum_{y\in S} \xi(y)\mathbb P_\theta[Y=y]$ coincides with the probability of a fuzzy subset in the sense of \cite{zadeh1968probability}, and it is interpreted as the degree of consistency of the fuzzy subset $\xi$ with respect to $\mathbb P_\theta$. 
\end{definition}

\begin{remark}
Unlike classical spaces of fuzzy numbers on $\mathbb R$ (e.g., normal convex fuzzy sets) where arithmetic is defined via $\alpha$-cuts \cite{lopez1997constructive}, $M\subseteq\mathcal B_b(S,\mathcal S)$ is used here only as a representation space: fuzzy subsets are identified with normalized $[0,1]$-valued functions. Hence $M$ is not closed under generic linear combinations, while it is closed under the pointwise supremum $\vee$. In our setting, no further geometric structure is needed.
\end{remark}

\begin{definition}
\textit{Granular count.} 
In a precise setting, the count of a referent $r$ (e.g., a gene expression) in a set $R$ emerges as the number $y \in S$ of observations $o$ in a set $O$ (e.g. the reads resulting from RNA-seq) that are assigned to the referent.
If observations are imprecise, the assignment is uncertain because they can be possibly assigned to more referents in $R$.
A possibilistic approach to counting enables deriving the possibility degree that a referent is assigned $y$ out of $K$ available observations, from the possibility degree $\pi_o(r)$ that an observation $o$ is assigned to referent $r$ \cite{mencar2020granular}.
The result is a fuzzy set $\tilde{y}$ with membership function:
$$
\xi_r(y) = 
    \max_{O_y \subseteq O}\lbrace
        \min\lbrace
            \min_{o \in O_y}{\pi_o(r)},
            \min_{o \notin O_y}{
                \max_{r' \in R\setminus\{r\}}{\pi_o(r')}
            }
        \rbrace
    \rbrace
$$
if $y\leq K$, and $\xi_r(y)=0$ if $y>K$.
The variable $O_y$ denotes a subset of $O$ with cardinality $|O_y|=y$ (by convention, $\min\emptyset=1$). Figure S.1 in the Supplementary Materials shows two exemplary cases of granular counts.
\end{definition}

\begin{remark}
Granular counting represents crisp counts as a collection of compatible alternatives rather than a single point, a concept formalized through fuzzy counts that weight candidates to create a graded compatibility structure. Unlike probability, which distributes a fixed total mass, possibility measures the degree to which alternatives remain compatible with observations: maximal possibility indicates an alternative cannot be excluded, while lower values reflect a progressive epistemic discounting of the candidate. Section S4 contains further theoretical and computational details, including a synopsis of possibility distributions, an efficient algorithm for granular counting, and a graphical illustration of exemplary counts (Figure S1).
\end{remark}

\newpage
\section{Fuzziness as Coarsening-Not-At-Random}\label{sec3}

This section states and discusses the main results supporting the view of fuzziness as a coarsening-not-at-random (CNAR) mechanism. 

\subsection{The statistical problem}\label{sec3_1}

Let $\{Y_i\}_{i=1}^n$ be a collection of $n$ independent $(\mathcal A,\mathcal S)$-measurable random variables, and let $\widetilde{\mathbf y}=\{\widetilde y_i\}_{i=1}^n$ denote the observed sample of fuzzy data. Because of epistemic uncertainty mechanisms, such as those acting on RNA-seq data \cite{o2015accounting}, $\widetilde{\mathbf y}$ can be viewed as a {imprecise version} of the unobserved vector of crisp realizations $\mathbf y=\{y_i\}_{i=1}^n$. Our goal is to model the associated blurring mechanism, which, after the latent outcome $Y(\omega)=y$ is generated, reports a fuzzy subset of $S$ rather than the natural (non-fuzzy) count $y$. Equivalently, we aim to perform inference on the parameter vector $\boldsymbol{\theta}$ indexing the joint distribution $f_{Y_1,\ldots,Y_n}(\mathbf y;\boldsymbol{\theta})$ given the fuzzy sample $\widetilde{\mathbf y}$.

\subsection{A Zadeh-oriented construction}\label{sec3_2}

In what follows, the finite case is adopted to keep the construction elementary in the discrete setting. Let $\Xi$ denote the fuzzy outcome modeled as an $(M,\mathcal M)$-valued random element. Conditionally on $Y=y$, $\Xi$ has distribution $\phi(y,\cdot)$, where $\phi:S\times\mathcal M\to[0,1]$ is a Markov kernel from $(S,\mathcal S)$ to $(M,\mathcal M)$, i.e. $\phi(y,A)=\mathbb P(\Xi\in A\mid Y=y)$ for $A\in\mathcal M$. In this setting, $\phi$ represents the fuzzy reporting mechanism. We also impose the support constraint $\phi\big(y,\{\xi\in M:\xi(y)>0\}\big)=1$ for all $y\in S$, so that outcomes incompatible with $y$ have zero probability. To exploit the fuzzy information $\xi$, let $\nu$ be a reference probability mass function on $M$ and define $c(y)\mathDef \sum_{\xi \in M} \xi(y)\,\nu(\xi)$, with $c(y)>0$ for all $y\in S$.\footnote{In this context, $\nu$ is a baseline distribution over the set of possible fuzzy reports $M$ and, in general, it plays the role of a prior over $M$.} Then set $\phi(y,A)\mathDef \frac{1}{c(y)}\sum_{\xi \in A} \xi(y)\,\nu(\xi)$, $A\in\mathcal M$. It is straightforward to show that for fixed $y\in S$, $\phi(y,\cdot)$ is a probability measure on $\mathcal M$ because it is a normalized finite sum of non-negative weights. Similarly, since $\mathcal S=\mathcal P(S)$, every function from $S$ to $\mathbb R$ is $\mathcal S$-measurable, hence $\phi(\cdot,A)$ is $\mathcal S$-measurable for each $A\in\mathcal M$. Note that the support constraint is inherently satisfied by this construction, as any $\xi$ such that $\xi(y)=0$ provides no contribution to the sum.

The proposed form of $\phi$ is rooted in three simple requirements: the reported fuzzy outcome should be compatible with the latent count $y$, reports assigning higher membership to $y$ should receive greater conditional weight, and unaccounted heterogeneity across admissible reports should be represented through a baseline distribution $\nu$. The kernel above is the simplest specification satisfying these requirements. In doing so, the generative link $y\mapsto\Xi$ explicitly incorporates the graded information encoded by $\xi$. Otherwise, the relation between the latent count and its fuzzy report would ignore the membership profile of $\xi$, effectively reducing to a set-valued coarsening scheme and squandering the added value of granular counts. A further technical argument in favour of this choice is that under this construction the marginal distribution of the fuzzy outcome $\mathbb P_\theta[\Xi \in A] = \sum_{y\in S}\phi(y,A)\mathbb P_\theta[Y=y]$ recovers the Zadeh probability of the fuzzy subset (see Definition 4). In particular, for a singleton $A=\{\xi\}$, the marginal is $\mathbb P_\theta[\Xi = \xi] = \nu(\xi)\sum_{y\in S}\frac{1}{c(y)}\xi(y)\mathbb P_\theta[Y=y]$. If $c(y)$ is constant in $y$ and $\nu$ is uniform on $M$, then $\mathbb P_\theta[\Xi = \xi] = \frac{1}{|M|c}C_\theta(\xi)$ is a Zadeh-type functional on fuzzy counts scaled by the factor $\frac{1}{|M|c}$. Notably, this allows for the fuzzy-event likelihood of \cite{gil1988operative} as a special case.

\subsection{The CNAR nature of the construction}\label{sec3_3}

We note that the general construction above generally entails a coarsening-not-at-random (CNAR) mechanism, in line with the characterizations in \cite{grunwald2003updating} and \cite{gill2008algorithmic} (a brief summary is provided in Section S1). 

More formally, consider $A=\{\xi\}$ and define the compatibility set $S_\xi \mathDef \{y\in S:\xi(y)>0\}$. We say that CAR holds for $\xi$ if $\phi(y,\{\xi\})=\phi(y',\{\xi\})$, for all $y,y'\in S_\xi$ (i.e., the probability of reporting $\xi$ does not depend on the specific value of $y$). However, under the Zadeh-oriented construction of Section~\ref{sec3_2}, the conditional probability of reporting $\xi$ varies with $y$ through the factor $\xi(y)/c(y)$. This immediately suggests that CAR typically fails whenever $\xi$ is not constant over $S_\xi$. 

\begin{proposition}[Characterization of outcome-wise CAR]\label{lem1}
Assume $\nu(\xi)>0$ and $c(y)>0$. Under the Zadeh-oriented construction, the mechanism is CAR in the singleton sense for the outcome $\xi$ if and only if $\xi(y)/c(y)$ is constant over $S_\xi$. 
\end{proposition}

\begin{proof}
Immediate from the definition of $\phi(y,\{\xi\})$.
\end{proof}

\noindent \textit{Example.} Let $S=\{0,1,2,3\}$, $M=\{\xi_1,\xi_2\}$, $\mathcal M=\mathcal P(\{\xi_1,\xi_2\})$, and take $\nu(\xi_1)=\nu(\xi_2)=\frac{1}{2}$.
Define the membership values by $\xi_1(0)=1,\ \xi_1(1)=\frac{1}{2},\ \xi_1(2)=\frac{1}{2},\ \xi_1(3)=\frac{1}{4}$ and $\xi_2(0)=\frac{1}{4},\ \xi_2(1)=\frac{1}{2},\ \xi_2(2)=1,\ \xi_2(3)=1$.
Then $c(y)=\frac{1}{2}(\xi_1(y)+\xi_2(y))$ and, for the singleton event $A=\{\xi_1\}$, the kernel gives $\phi(y,A)=\frac{\xi_1(y)}{(\xi_1(y)+\xi_2(y))}$. The compatibility set is $S_{\xi_1}=S$.
In particular, $\phi(0,\{\xi_1\})=\frac{4}{5}$ while $\phi(3,\{\xi_1\})=\frac{1}{5}$. Hence CAR fails.\\

This characterization clarifies why CAR is exceptional under fuzzy reporting. Indeed, once the reporting mechanism genuinely exploits graded membership, the resulting coarsening mechanism is typically non-ignorable. In this sense, CNAR is not a pathological feature of the proposed construction, but the generic consequence of linking fuzzy reports to latent counts through their compatibility profile. The inferential implication is immediate: when reporting is non-ignorable, inference on $\boldsymbol\theta$ cannot rely on the latent count model alone. Rather, as in MNAR models \cite{molenberghs2005models}, one must specify the measurement model $Y\sim\mathbb P_\theta$ together with the coarsening mechanism $\Xi\mid(Y=y)\sim \phi(y,\cdot)$. This is the rationale for the hierarchical model developed in Section \ref{sec3_4}.

\subsection{A CNAR model instance}\label{sec3_4}

We now specialize the general construction to a Beta-type parametric family $\xi_{c,h}$ of fuzzy sets, which will later be used in the RNA-seq application. Let $M_\text{be} \subset M$ denote the class of Beta-type possibility functions. Each fuzzy outcome $\xi$ is parametrized by two coordinates $(c,h)$, where $c\in[0,K]$ and $h>0$. Let $\Omega_M=[0,K]\times(0,\infty)$ and let $\eta:\Omega_M\to M_\text{be}$ denote the deterministic map $(c,h)\mapsto \xi_{c,h}$. Conditionally on the latent count $Y_i\sim F_{Y_i}(y;\boldsymbol\theta)$, the coordinates of a fuzzy outcome are generated as follows: (i) $ H_i\sim \mathcal Ga(\alpha_h,\beta_h)$, (ii) $C_i\mid H_i,Y_i \sim \mathcal Be (h_i\bar y_i,\,h_i-h_i\bar y_i)$ with $\mathcal Ga$ being the Gamma distribution (rate parametrization), $\mathcal Be$ the Beta distribution, and $\bar z=z/K$. The observed fuzzy outcome is then $\Xi_i=\eta(K C_i,H_i)$.\footnote{This coordinate-based specification separates the pure aleatory component from the epistemic fuzziness mechanism. The conditional Beta law yields flexible, possibly skewed reports while keeping an explicit link between $(c,h)$ and $y$. Moreover, $h$ controls limiting regimes: large $h$ concentrates the report around $y$ (crisp limit), whereas small $h$ produces diffuse reports, suggesting that defuzzification may distort dispersion-related inference. See \cite{calcagni2025bayesianize} for details.} Further details are provided in Section S2.2.

\section{Case study}\label{sec4}

We now return to the motivating RNA-seq setting and illustrate and evaluate the proposed framework through a real case study. 

Data refer to $n=89$ RNA-seq samples from high-throughput sequencing of human pancreatic islets (GSE50244), generated on the Illumina HiSeq 2000 platform and originally analyzed to investigate genes influencing glucose metabolism \cite{fadista2014global}. Raw sequences were processed using the STAR aligner and RSEM, and the resulting transcriptomes were subsequently analyzed with the MultiDEA method for uncertainty quantification \cite{consiglio2016fuzzy} (see Sections S6.1--S6.2). Among the sequenced genes, we focus on HAS3 as an illustrative case study. Much like functional data analysis represents functional objects using low-dimensional basis representations, the raw granular counts for HAS3 (i.e., the observed possibility distributions) were approximated by Beta-type fuzzy sets (see Section S2.2), yielding a tractable parametric representation (see Section S5 and Section S6.3a). The final outcome variable consists of $n=77$ paired observed statistics $\{(c_i,h_i)\}_{i=1}^n$ from the original fuzzy counts (only complete cases were retained), which constitute the input data for the subsequent analyses. 

In this application, $F_{Y_i}(y;\boldsymbol\theta) \mathDef \mathcal{N}eg\mathcal{B}in(y;\mu_i,\kappa)$, where $\mu_i = u_i\exp\{\mathbf z_i\boldsymbol\beta\}$ is the linear predictor connecting the vector of covariates $\mathbf z_i$ to $\mathbb E[Y_i]$ scaled by the normalization factor $u_i$ (i.e., offset of the model) and $\kappa>0$ is the gene-specific dispersion parameter. Inference on $\boldsymbol\theta=\{\boldsymbol\beta,\kappa\}$, based on the observed statistics $\{(c_i,h_i)\}_{i=1}^n$, is carried out in a Bayesian setting via Hamiltonian Monte Carlo (see Section S6.4). To explore covariates associated with HAS3 expression, we considered a small hypothesis-driven set of models: a null model (M0), a model with HbA1c only (M1), a model adding BMI, age, and biological sex (M2), and a model further including the interaction HbA1c $\times$ biological sex (M3). Predictive comparison via PSIS-LOO-CV and WAIC selected M1, which is therefore used in the analyses below (technical details, diagnostics, and posterior summaries are reported in Sections S6.3b-e).

In what follows, we focus on how our framework relates to more general standard approaches. Since the competing approaches operate on different representations of the data, we study two distinct modeling choices. First, what is lost when fuzzy counts are compressed into point-valued proxies commonly used after RNA-seq quantification (Section 4.1). Second, what is lost when fuzziness is retained, but the reporting mechanism is treated as ignorable (Section 4.2). 

\subsection{Compressing fuzziness by scalar proxies}\label{sec4_1}

Here we examine what is lost when the observed fuzzy counts are replaced by external scalar proxies. This is obtained through defuzzification, by replacing $\{(c_i,h_i)\}_{i=1}^n$ with scalars $\{\bar c_i\}_{i=1}^n$ \cite{calcagni2024estimating}. At this representation level, the closest approach for comparison is provided by RSEM, which calculates expected counts using the EM algorithm \cite{li2011rsem}. Defuzzified and expected counts (see Figure S5) were used as input of the M1 model specification, whose parameters were estimated as previously done (see Table S2). The main discrepancy concerns the dispersion component, whereas the regression structure is comparatively less affected. In particular, the RSEM-based specification yields a smaller posterior mean for $\kappa$ ($\hat\kappa_{\text{RSEM}} = 1.66$) than the defuzzified specification ($\hat\kappa_{\text{Defuzz}} = 2.77$), together with a more concentrated posterior distribution. Thus, among scalar summaries, the defuzzified counts produce dispersion estimates that are larger and more uncertain than those obtained from the RSEM-based proxy. Relative to CNAR, these scalar-proxy approaches mostly differ in how they estimate the dispersion parameter $\kappa$ and quantifies its posterior uncertainty. The result is not surprising and it is in line with other findings \cite{calcagni2024estimating}, suggesting that defuzzification mainly affects second-order inference rather than first-order regression structures (see Section S6.3f). To explore whether defuzzification offers advantages over RSEM, we also performed a posterior predictive check (PPC) \cite{gelman1996posterior} and compare two posterior statistics, namely the scaled mean and the 80\% scaled IQR (see Section S6.3f). The results indicate that compared to RSEM, defuzzified counts maintain a closer link to the granular data, indicating that defuzzification, albeit reductive, still captures some features of the observed data (see Figure S6). 

\subsection{Retaining fuzziness under ignorability}\label{sec4_2}

We now keep the observed data in their original fuzzy form and examine the effect of ignoring the conditional reporting mechanism $\Xi \mid Y=y \sim \phi(y,\cdot)$ through a CAR-like specification. In particular, we consider two alternatives to CNAR, denoted CAR$^1_{\text{like}}$ and CAR$^2_{\text{like}}$ (see Section S3), and compare them by posterior predictive analyses. Figure \ref{fig1A} shows the predictive distributions for $(c,h)\in\Omega_M$ under the three model instances (see also Table S5). CAR$^1_{\text{like}}$ yields a predictive distribution that is too concentrated relative to the observed data, coherently with the fact that this formulation does not explicitly represent variability in the latent count. CAR$^2_{\text{like}}$ partly compensates for this through the additional dispersion parameter $\lambda\in(0,\infty)$, but the correction remains largely global: the resulting predictive distribution still resembles a single Beta-like shape with a sharper peak and thinner shoulder regions than under CNAR. In this sense, the CAR-like models can recover some broad location/dispersion features, but they reproduce the observed fuzzy sample less faithfully at the level of the joint $(c,h)$ structure. We therefore complement this comparison with an energy-like analysis at the level of the full fuzzy outcomes $\xi_{c,h}\in M_{\mathrm{be}}$ (see Section S6.3g). The aim is to assess whether the model reproduces the internal structure of the observed fuzzy sample, rather than only a few marginal summaries. To this end, we compare the discrepancy measures $u_{\mathrm{rep}}$ and $u_{\mathrm{cross}}$ with $u_{\mathrm{obs}}$: values of $u_{\mathrm{cross}}$ close to $u_{\mathrm{obs}}$ indicate that replicated fuzzy counts are structurally compatible with the observed sample. Figure \ref{fig1B} shows that CNAR yields replicated samples more closely aligned with $u_{\mathrm{obs}}$, whereas the CAR-like alternatives remain systematically farther away. Thus, treating fuzziness as ignorable may preserve some coarse features of the data, but it distorts the local shape and within-sample structure of the observed granular counts.

\begin{figure}
\hspace*{-1cm}
    \begin{subfigure}[b]{0.48\textwidth}
        \centering
        \includegraphics[width=1.15\textwidth]{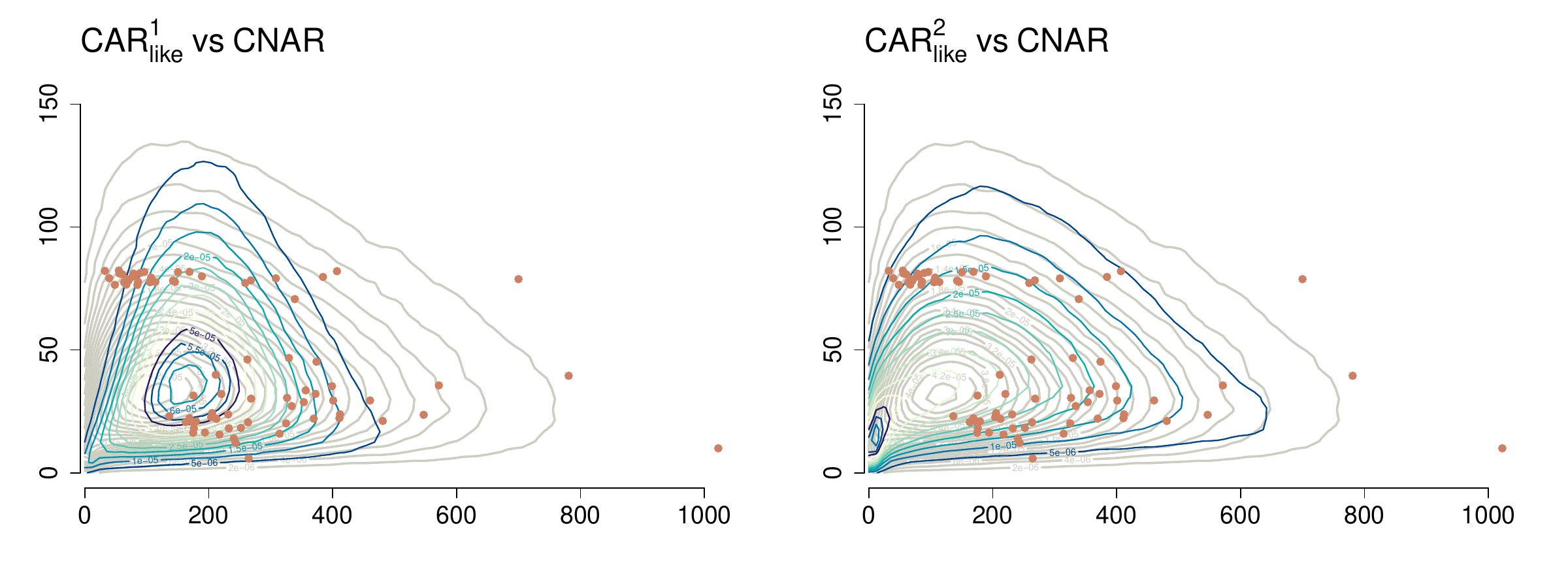}
        \caption{Comparison of observed vs. predicted $(c,h)\in\Omega_M$.}
        \label{fig1A}
    \end{subfigure}
    \hfill 
    \begin{subfigure}[b]{0.48\textwidth}
        \centering
        \includegraphics[width=1.15\textwidth]{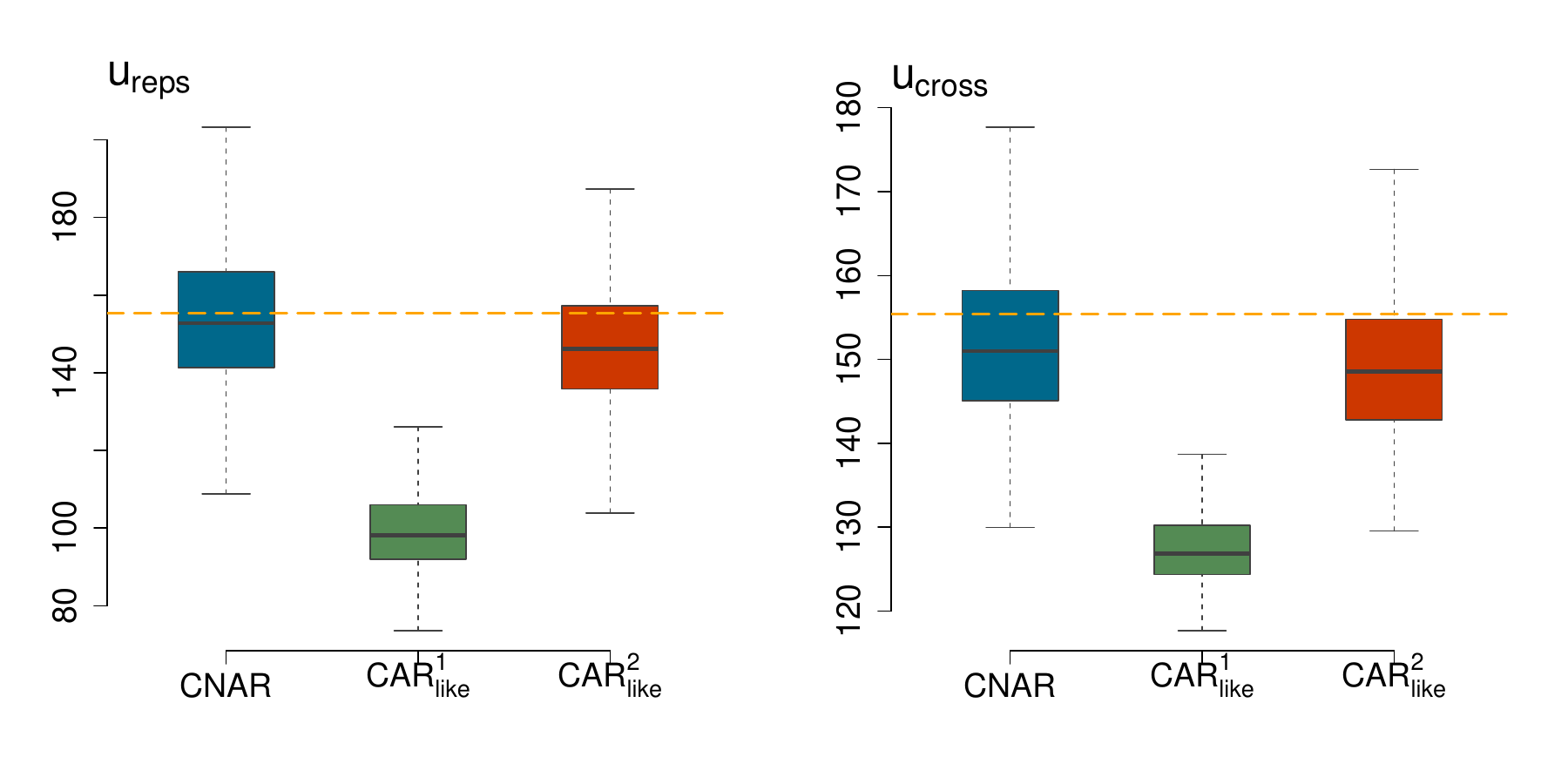}
        \caption{Comparison in terms of energy-{like} components (dashed lines represent $u_{\text{obs}}$).}
        \label{fig1B}
    \end{subfigure}
    
    \caption{Case study: Comparative analysis between CNAR and CAR-like model instances.}
    \label{fig1}
\end{figure}

\subsection{Results in brief}\label{sec4_3}

Taken together, the empirical comparisons point to a consistent pattern. When fuzzy counts are compressed into scalar proxies, the main inferential loss concerns dispersion and its uncertainty rather than the first-order regression signal: relative to RSEM expected counts, defuzzified counts remain closer to the observed granular data and yield less compressed dispersion inference. When fuzziness is retained but treated as ignorable, the loss shifts from scalar dispersion summaries to the structure of the fuzzy sample itself: the CAR-like formulations reproduce some broad features of the data, but they are less successful than CNAR at matching the observed joint $(c,h)$ distribution and the internal structure of the fuzzy counts. The practical takeaway is that ignoring fuzziness matters less for the regression trend and more for capturing the right dispersion and the structural link between observed and replicated outcomes.

\section{Conclusions}\label{sec5}

In this paper, we argue that fuzziness in granular counts is not an artifact, but rather the result of an informative coarsening process. Motivated by RNA-seq data, the main theoretical contribution of the paper is developed in Section~\ref{sec3}, where we introduce a general class of fuzzy-reporting mechanisms based on a Zadeh-oriented use of graded membership. A central implication of this construction is that ignorability fails generically. Indeed, since the probability of observing a given fuzzy report typically depends on the latent outcome itself, the mechanism is coarsening-not-at-random, except in special cases. We believe that this point is of crucial importance: the non-ignorability of the data arises already at the level of granular counting. In this sense, the hierarchical representation reveals CNAR to be a logical consequence of granular counts once graded compatibility is included in the model specification.

\clearpage
\bibliographystyle{plain}
\bibliography{biblio}

\end{document}